\documentclass[11pt, a4paper]{article}
\usepackage[margin=2.54cm]{geometry}
\usepackage[utf8]{inputenc}
\usepackage[numbers]{natbib}
\usepackage{amsmath, amssymb, amsfonts, amsthm}
\usepackage{color}
\usepackage{mathtools}
\usepackage[hidelinks, pdfencoding=auto, psdextra]{hyperref}
\usepackage{graphicx}
\usepackage{verbatim}
\usepackage[width=.95\textwidth]{caption}
\usepackage[nameinlink,noabbrev]{cleveref}
\usepackage{jscommands}
\usepackage[section]{placeins}
\newcommand{\flatbush}{\texttt{Flatbush}}
\newcommand{\splitp}{\texttt{SplitP}}

\begin{document}
\title{\vspace{-15mm}
Evaluation of the relative performance of the subflattenings method for phylogenetic inference}
\author{Joshua Stevenson, Barbara Holland, Michael Charleston, Jeremy Sumner}
\maketitle
\begin{abstract}
    The algebraic properties of \textit{flattenings} and \textit{subflattenings} provide direct methods for identifying edges in the true phylogeny---and by extension the complete tree---using pattern counts from a sequence alignment. The relatively small number of possible internal edges among a set of taxa (compared to the number of binary trees) makes these methods attractive, however more could be done to evaluate their effectiveness for inferring phylogenetic trees. This is the case particularly for subflattenings, and our work makes progress in this area. We introduce software for constructing and evaluating subflattenings for splits, utilising a number of methods to make computing subflattenings more tractable. We then present the results of simulations we have performed in order to compare the effectiveness of subflattenings to that of flattenings in terms of split score distributions, and susceptibility to possible biases. We find that subflattenings perform similarly to flattenings in terms of the distribution of split scores on the trees we examined, but may be less affected by bias arising from both split size/balance and long branch attraction. These insights are useful for developing effective algorithms to utilise these tools for the purpose of inferring phylogenetic trees.
\end{abstract}

{\small \tableofcontents}

\section{Introduction}

The accurate inference of species trees from aligned sequence data is a key task that underpins much of evolutionary biology. 
However, it is a very challenging statistical problem. 
Much of this difficulty is because sequence evolution is a highly heterogeneous process. 
The substitution process can vary in different parts of the evolutionary tree, and also in different parts of the genome. 
In addition to this, the local gene tree along which sequences evolve may vary across the alignment. 

There have been two broad approaches in trying to cope with this complexity. 
The first approach is to model it; likelihood-based and Bayesian approaches have become increasingly parameter rich as they use partitions and/or mixtures to account for heterogeneity. 
The second approach has been to develop methods that are robust to this kind of complexity and avoid having to estimate what are essentially large numbers of nuisance parameters. 
For example, the \textit{logDet} method (see \citepair{Lockhart1994}) allows for consistent inference under the General Markov Model (GMM) without having to fit any parameters.
Methods based on \textit{phylogenetic invariants} also fall into this class of approaches. 
While invariant methods were not initially competitive compared to maximum likelihood \cite{Huelsenbeck1993}, there has been a renaissance in their use over the last decade or so. Explicit use of algebraic methods in phylogenetics started in the late 1980s with both the \textit{Hadamard conjugation} ideas of Hendy and Penny \cite{Hendy1989} and the discovery of \textit{phylogenetic invariants} by Cavender, Felsenstein and Lake \cite{Cavender1987,Lake1987}. These initial ideas exploited the underlying symmetries inherent in both the relevant molecular substitution models, e.g, the Kimura 3ST model \cite{kimura1981estimation}, and the phylogenetic tree itself (leaf permutation symmetries). Importantly, the theory of phylogenetic invariants eventually led to the comprehensive application of methods from algebraic geometry, beginning with the work of \citeauthor{Allman2004} \cite{Allman2004,Allman2007}. A breakthrough paper by \citeauthor{Allman2008} in \citeyear{Allman2008} \cite{Allman2008} described a new class of invariants based on the idea of a (tensor) \textit{flattening} and the rank properties of matrices of site pattern probabilities.

The \textit{flattening} construction is quite natural to phylogenetics as it corresponds to identifying evolutionary splits (partitions of the set of taxa into two parts, akin to an edge in a tree) by reinterpreting a multiple sequence alignment to consist of two sequences of expanded characters (one sequence for each side of the split).
For a split $A|B$ (on a DNA alignment of $N=|A|+|B|$ sequences), the sequence alignment can be summarised into a $4^{|A|}\times 4^{|B|}$ \emph{matrix} with entry $I,J$ being the count of alignment locations with subpattern $I$ on the left side of the split and subpattern $J$ on the right.
The biological intuition: if one takes a split corresponding to an edge in the true tree, the probability mass contained in the flattening matrix will be concentrated on certain off-diagonal entries, with the converse holding for splits not corresponding to edges in the tree.
This information can then be detected by measuring the (approximate) algebraic rank of the flattening matrix, which, in particular, is typically done using the singular value decomposition (\Cref{sec:background}).
Thus, the flattening matrices provide a very direct means of detecting phylogenetic signal and identifying edges in the tree.

Subsequent papers by \citepair{Chifman2015} and \citepair{Gaither2016} observed the flattening rank properties in practice, using sequence data generated under the multi-species coalescent and with a distribution of rates across sites. 

Using a different perspective and starting point, the work by \citeauthor{sumner2005entanglement} \cite{sumner2005entanglement, Sumner2008Markov} on \textit{Markov invariants} has also taken an algebraic approach to phylogenetics but has instead concentrated on the time evolution of the underlying Markov process. In particular, \citepair{sumner2005entanglement} detail a particular set of Markov invariants, the \textit{squangles}, which can be used for quartet reconstruction. We use squangles as a point of comparison alongside flattening and subflattening methods in \Cref{sec:analysis} following the process suggested by \citepair{Holland2012}. We direct the reader to the two aforementioned papers for background on squangles and how they may be used to evaluate splits.

The general philosophy followed here is that, under standard phylogenetic models, evolutionary divergence events are modelled as discrete instantaneous branching of lineages, whereas the subsequent evolution occurs as a continuous time random process of nucleotide substitutions.
Thus, if one views the primary difficulty of phylogenetic inference to be the correct identification of branching events (i.e., the inference of the phylogenetic tree), the random nucleotide substitutions that occur on the pendant edges of the tree contribute nothing and make the inference more difficult.   
Indeed, the above stated properties of the flattening matrix are contingent on a particular edge being identifiable with evolutionary divergence events occurring some finite time prior.
Thus it is valuable to understand the theoretical behaviour of the flattening matrix as time passes under the given Markov model of sequence evolution.
The relevant matrix transformation rule is well understood (and was presented in Allman Rhodes \cite{Allman2008}), and is crucial to the the theoretical derivation of the rules for the rank of the flattening matrices.
Exploring these ideas further, Sumner \cite{Sumner2017} showed analogous transformation rules are obtainable from `subflattenings', which are particular submatrices of the flattening matrices.
These submatrices have significantly smaller dimensions than the flattenings, reducing the dimension from $k^{|A|}\times k^{|B|}$ to $({|A|}(k-1)+1)\times ({|B|}(k-1)+1)$ for a split $A|B$, where $k$ is the size of the state space (for example $k=4$ for the state space of DNA nucleotides $\kappa := \set{A,C,G,T}$).
Thus the subflattening has the advantage of having dimensions which are linear in the size of each split, and \citeauthor{Sumner2017} showed that the rank conditions on the subflattenings matrices can also be used to (theoretically) identify splits corresponding to edges in the true phylogenetic tree \cite{Sumner2017}.

The purpose of the present work is to show the relative performance of these methods---in terms of their ability to identify edges in the true tree---under a simple simulation framework. In particular, we focus on simulating under the Jukes Cantor model with varying sequence length and branch lengths. Throughout, `branch length' refers to the expected number of substitutions per site. We also do not consider complications arising under the coalescent, in order to keep the presentation and comparisons as clean and clear as possible. First, we provide the relevant background definitions and theorems, including the flattening and subflattening rank theorems. In \Cref{sec:code_and_computation}, we identify some important considerations in the implementation of flattenings and subflattenings, and describe the software we have developed to construct such matrices and identify true splits. In \Cref{sec:analysis}, we detail the simulations we performed in order to evaluate the performance of subflattening matrices in comparison to flattenings and squangles. The results of the simulations are discussed in \Cref{sec:results}, and we conclude with a discussion in \Cref{sec:discussion} including suggestions for future work.

\section{Background} \label{sec:background}

For a given set of taxa $X$, a \textit{split} is defined as \textit{partition} of $X$ into two sets $A$ and $B$ and is written $A|B$ (or equivalently $B|A$). We may also omit set notation. For example, if $X=\{1,2,3,4\}, A=\{2\}$ and $B=\{1,3,4\}$, then the split $A|B$ can be written $134|2$. Since splits are bipartitions, we can also omit one side of the split when $X$ is clear. In this case, we might write $134|2$ as $134$ or simply $2$. We define the \textit{size} of a split to be the smallest subset (for example, the size of $134|2$ is 1) and note that the size of a split can be thought of as its \textit{balance}; this is because for a fixed number of taxa, the larger the size of a split, the more balanced the taxa are across the two subsets. Splits are a natural idea in phylogenetics, because each edge in a phylogenetic tree can be identified with a split. To see this, notice that removing an edge from a tree creates two smaller trees and hence a partition of the taxa into two sets; the induced partition can be represented as a split. Of course, not all splits will represent edges in a given tree. For example, if a tree displays two taxa as a cherry, then any split with size more than 1 which separates these two taxa will not correspond to any edge in the tree. Splits with size equal to $1$ will always correspond to leaf edges, and for this reason they are referred to as \textit{trivial} splits. Given a tree, we say that splits which correspond to edges are \textit{displayed} by the tree, and refer to them as $true$ splits. For example, trivial splits are always true splits. Splits which are not displayed by the tree are referred to as false splits. Splits are \textit{compatible} if there exists some tree which can display both splits. For example, the splits $01|23$ and $02|13$ are not compatible on any tree. Due to the correspondence between true splits and edges, the complete set of true splits uniquely defines the tree \cite[Proposition 2.4]{Steel2016}. 

Fundamentally, inference of phylogenetic trees amounts to the identification of true splits. Importantly, there are far fewer splits on $n$ taxa than there are trees on $n$ taxa. The number of splits is given by $2^{(n-1)}\!-\!1$ whereas the number of unrooted trees is given by $(2n-5)!! = \frac{(2 (k-2)) !}{2^{(k-2)} (k-2) !}$ (which grows like $2^k k!$). It is also worth noting that a tree without branch lengths (or with branch lengths which are all equal) may have a number of splits which are symmetrically equivalent (that is, they are invariant under an action of the automorphism group of the tree under vertex permutations). For example, for a balanced four-taxon unweighted tree with leaf labels $\{0,1,2,3\}$ and cherries $\{0,1\}$ and $\{2,3\}$, we can say that the (false) splits $02|13$ and $03|12$ are symmetrically equivalent. A true split on a tree induces two sub-trees which are disjoint (that is, they share no edges in the original tree). A false split, however, will always give rise to two sub-trees which share one or more internal edges in the tree. The number of shared edges between the two sub-trees may be taken as a measure of `how false' a given false split is on a tree, and is a statistic which explains some of the results displayed in the next section.

In practice, we can obtain a measure of how likely it is (in a colloquial sense) for a split to be true, using the rank properties of flattening and subflattening matrices. The flattening rank properties are presented in the following theorem.

\begin{theorem}
	\label{thm:flat_rank}
	Given a phylogenetic tree $T$ and state space $\kappa = \set{1,...,k}$, the \textit{flattening} matrix $\mathrm{Flat}_{A|B}(P)$ computed from a split $A|B$ has $\mathrm{rank}(\mathrm{Flat}_{A|B}(P)) = k$ if the split $A|B$ is displayed by $T$, and has $\mathrm{rank}(\mathrm{Flat}_{A|B}(P)) = k^{p_{A|B}}$ otherwise, where $p_{A|B}$ is the parsimony score of the split $A|B$ on the tree $T$.
\end{theorem}
\begin{proof}
A proof is provided by \citepair{Snyman2021}, strengthening a previous result by \citepair{Eriksson2005} and \citepair{FernanndezSanchez2015}.
\end{proof}

We will discuss shortly how the rank of these matrices can be evaluated in practice, but first we note that these matrices are exceedingly large. The subflattening is a similar construction which is much smaller in dimension \cite{Sumner2017}. To understand the theoretical origin of subflattenings, we give an illustrative explanation using the simplest possible scenario: binary sequence data ($k=2$) and a split $A|B$ with $|A|\!=\!|B|\!=\!2$.
The general case is presented in \cite{Sumner2017}.

For the binary states, one can express a general $2\times 2$ Markov matrix as
\[
M=
\left(
\begin{matrix}
1-a & a\\
b & 1-b
\end{matrix}
\right).
\]
Following the well-known Hadamard conjugation \cite{hendy1994discrete}, we consider what happens to this matrix under conjugation by a Hadamard matrix $H$:
\[
    HMH^{-1}=
    \left(
    \begin{matrix}
        1 & 1\\
        1 & -1
    \end{matrix}
    \right)
    \left(
    \begin{matrix}
        1-a & a\\
        b & 1-b
    \end{matrix}
    \right)
    \left(
    \begin{matrix}
        \tfrac{1}{2} & \tfrac{1}{2}\\
        \tfrac{1}{2} & -\tfrac{1}{2}
    \end{matrix}
    \right)
    =
    \left(
    \begin{matrix}
        1 & b-a\\
        0 & 1-a-b
    \end{matrix}
    \right)
    \equiv
    \left(
    \begin{matrix}
        1 & v\\
        0 & z
    \end{matrix}
    \right),
\]
say. Note that the constraint $0 < a,b < 1$ is converted to $-1 < v,z < 1$ and 
Now, multiplying two matrices of the above form gives
\[
    \left(
    \begin{matrix}
        1 & v\\
        0 & z
    \end{matrix}
    \right)
    \left(
    \begin{matrix}
        1 & v'\\
        0 & z'
    \end{matrix}
    \right)
    =
    \left(
    \begin{matrix}
        1 & vz'+v'\\
        0 & zz'
    \end{matrix}
    \right),
\]
or, directly in terms of parameters, we have:
\[
    (v,z)\ast (v',z') := (vz'+v',zz').
\]
This transformation rule tells us algebraically exactly what is happening to the substitution parameters on a single branch of a phylogenetic tree as time passes.

More generally, if we consider a phylogenetic tree and concentrate on the two lineages and substitutions between the the $2^2\!=\!4$ pairs of binary states as time passes, algebraically the independence of substitutions on the two lineages can be incorporated by taking the tensor product of two Markov matrices $M_1\otimes M_2$.
Working with the Hadamard conjugated matrices, this gives us (via the rules for Kronecker products):
\[
M_1\otimes M_2 \mapsto (HM_1H^{-1})\otimes (HM_2H^{-1})=
\left(
\begin{matrix}
1 & v_1\\
0 & z_1
\end{matrix}
\right)
\otimes 
\left(
\begin{matrix}
1 & v_2\\
0 & z_2
\end{matrix}
\right)
=
\left(
\begin{matrix}
1 & v_2 & v_1 & v_1v_2\\
0 & z_2 & 0 & v_1z_2\\
0 & 0 & z_1 & z_1v_2\\
0 & 0 & 0 & z_1z_2
\end{matrix}
\right).
\]
It is then a tedious exercise (left to the reader) to verify that multiplying two such $4\times 4$ matrices produces the following transformation rule stated in terms of the substitution parameters:
\begin{align*}
    (v_1,z_1)\ast (v_1',z_1') &\mapsto (v_1z_1'+v_1',z_1z_1'), \text{and}  \\
    (v_2,z_2)\ast (v_2',z_2') &\mapsto (v_2z_2'+v_2',z_2z_2').
\end{align*}
The key idea presented in \cite{Sumner2017} was that the same information is contained in upper-left $3\times 3$ submatrix  
\[
\left(
\begin{matrix}
1 & v_2 & v_1 \\
0 & z_2 & 0 \\
0 & 0 & z_1 
\end{matrix}
\right),
\]
as, again, the reader can verify by multiplying two matrices of this type.
The subflattening itself is then obtained by (i) considering a four taxon alignment of binary pattern counts, (ii) forming the corresponding $2\times 2 \times 2 \times 2$ array of binary state pattern counts, (iii) implementing the Hadamard conjugation on this array, (iv) choosing a split and flattening the array of pattern counts into the corresponding $4\times 4$ matrix, and (v) taking upper-left $3\times 3$ submatrix.  

\begin{theorem} \label{thm:subflatRankPars}
	Let $T$ be a tree with site pattern probability distribution $P$. Let $A|B$ be a split. A subflattening matrix $\mathrm{Subflat}_{A|B}(P)$ has rank $k$ if $A|B$ is displayed by $T$, and rank $r(k\!-\!1)\!+\!1$ otherwise, where $r$ is the parsimony score of the split $A|B$.
\end{theorem}
\begin{proof}
    See \citepair{Sumner2017}.
\end{proof}

Generalising to an arbitrary number of lineages leads directly to the subflattening matrices, but the above gives the basic underlying idea. Note that for the remainder of this paper, we will assume that $k=4$, referring to the size of our state-space of DNA nucleotides $\kappa:=\set{A,C,G,T}$.

As stated above, the dimensions of subflattenings is greatly reduced in comparison to the flattenings (exponential to linear).
However, this comes at the price of the Hadamard transformation which means that the entries of the subflattenings are not directly obtainable from counting patterns in the sequence alignment. That is, entries do not correspond to single site-patterns, as they do for flattenings. Despite this, the subflattenings can still be computed efficiently, as described in \Cref{sec:code_and_computation}.

After flattenings and subflattenings have been constructed using observed pattern frequencies from a sequence alignment, extracting information from them involves evaluating how `close' the given matrix is to having rank $k$. Singular value decomposition has often been suggested in the literature as a way to evaluate this \cite{Allman2017, Eriksson2005}, due to the connection between singular values and low-rank matrix approximation \cite{Eckart1936}. Put briefly, the sum of the singular values---excluding the $m$ largest---of a matrix $F$, is the distance between $F$ and the closest rank-$m$ matrix to $F$ under the Frobenius norm. We are interested in the case where $m$ is set to the number of states, $k$. \citeauthor{Allman2017} \cite{Allman2017} define the `split score':

\begin{definition}
For flattening (or subflattening) matrix $F_{A|B}$ arising from a split $A|B$ and a sequence alignment with $k$ possible states, define the \textit{split score},
\[ \displaystyle \mathrm{SplitScore}\!\left(F_{A|B}\right)=\left(1-\frac{\sum_{i=1}^{k} \sigma_{i}^{2}}{\|F_{A|B}\|^{2}}\right)^{\frac{1}{2}}, \]
where $\sigma_{1} \geq \sigma_{2} \geq \cdots \geq \sigma_{\min (m, n)} \geq 0$ are the singular values of $F_{A|B}$, and $\|F_{A|B}\|$ is the Frobenius norm of $F_{A|B}$.
\end{definition}

The split score is a normalised measure of the distance to the nearest rank-$k$ matrix and requires only the largest $k$ singular values to be computed. This saves time on computation. The code we provide to evaluate splits using flattenings and subflattenings utilises this definition of the split score.
\section{Code and computational considerations} \label{sec:code_and_computation}

\subsection{Code}

Two separate programs were developed for constructing subflattenings, computing split scores and performing the simulations and analysis detailed in \Cref{sec:analysis}. The first implementation, \flatbush\ was written in C++ by \citepair{Charleston2021}. \flatbush\ is a command-line program that accepts either FASTA or NEXUS format sequence alignment data, currently limited to the nucleotide alphabet of characters, $\kappa = \set{A,G,C,T}$.
If the NEXUS format is used then a ``FLATBUSH" block may be included, which may include instead of an alignment a set of splits with weights (which would otherwise be calculated from an alignment).
Splits may be input either as taxon names or in a compact binary encoding, and can be entered explicitly or by opting to do all possible non-trivial splits, subject to the number of taxa not being too large (in this case 16, however note that in principle \flatbush\ can accept any number of taxa and sequences as input). Input arguments may also include an input tree in simple Newick format, or a set of splits of interest. Alternatively, a tree or set of splits may be provided as command-line arguments. Please see the supplementary material for additional information.

A second implementation, \splitp\ is a python package developed by \citepair{Stevenson2022}. As \splitp\ was written independently of \flatbush, we were able to verify our calculations by comparing outputs of both programs with the same inputs. While \splitp\ is generally slower than the \flatbush\ implementation, it still includes the optimisations detailed in the remainder of this section, and can also compute and evaluate scores for flattening matrices and squangles. For this reason, we use \splitp\ for the analyses in \Cref{sec:analysis}. As input, \splitp\ accepts FASTA formatted sequence alignments or a table of site-pattern counts. Alternatively, \splitp\ also accepts trees with up to 36 taxa in the Newick format. \splitp\ can compute flattenings and subflattenings using exact site-pattern probabilities, or can generate sequence alignments of a given length, making it useful for simulations detailed in \Cref{sec:analysis}. Any sub-model of the general Markov model is supported, as the user can reassign any Markov matrix to each edge in the tree. \splitp\ is available on PyPI as well as on GitHub (at \url{https://github.com/js51/SplitP}) as open source software, as is \flatbush\ (at \url{https://github.com/mcharleston/Flatbush}) and we encourage code contributions.

\subsection{Sparsity}
Naïve approaches to problems involving very large matrices quickly lead to memory problems. For example, with just 8 taxa, the number of entries in a flattening matrix is $4^{16} \approx$ 4.2 billion. Assuming the entries are 4 bytes each, this leads to over 17 gigabytes of RAM required just to store the flattening in memory. Adding just one more taxon increases this number to more than 4.3 terabytes. Fortunately, flattening matrices are sparse. That is, the number of non-zero entries is typically in the order of the number of rows or columns. There exist a number of computational techniques for dealing with large sparse matrices, including techniques for computing the singular value decomposition, for example. These techniques are implemented in most major scientific computing packages, for example SciPy \cite{scipy} for Python or Eigen \cite{eigenweb} for C++.

\textit{Subflattenings} on the other hand are never sparse, but have dimensions which are quadratic in the number of taxa rather than exponential, and can therefore offer an even greater benefit than sparse matrix methods in certain situations. In a `worst case' scenario, in which each site pattern is equally likely to be observed, the expected number of non-zero entries in a flattening matrix is $L\!-\!{(L-1)^L}{L^{1-L}}$, where $L$ is the sequence length. In comparison, the expected number of non-zero entries in the subflattening is $(3t+1)^2$, where $t$ is the number of taxa. We can therefore determine under this worst-case scenario, the sequence length required to benefit (in terms of memory use) from subflattenings for a given number of taxa. For example, given a $6$-taxon tree, a sequence length of around $570$ is required, whereas for a $15$-taxon tree, a sequence length of over $3000$ is required in order to derive benefit from subflattenings. Of course, site-pattern probability distributions are never uniform, and so in practice these thresholds are much higher. 

\subsection{The Singular Value Decomposition}
Due to the ubiquity of the singular value decomposition in numerical methods, packages for computing singular values are heavily optimised. This is true even for the sparse matrix SVD packages, which \citepair{Allman2017} employ in order to compute split scores for flattenings. \splitp\ utilises SciPy's sparse matrix methods to compute singular values for flattenings, and standard methods for subflattenings. \flatbush\ utilises the C++ package Eigen \cite{eigenweb}. Importantly, the definition of the split score means that only the $k$ largest singular values are needed, leading to significant performance improvements for both methods.

\subsection{Subflattening transformation}
The bulk of the computation time for constructing the subflattenings from site pattern counts comes from the application of the following transformation of the (theoretical or observed) site-pattern distribution $P$ on the set of taxa $X$, with $|X|=t$.
\[ P \rightarrow \underbrace{H \otimes H \otimes \ldots \otimes H}_{t \text{ times}} \cdot P. \]
This transformation is equivalent to the one described in \Cref{sec:background}.
Specifically, the entries of the subflattening all come as sums of site-pattern probabilities,
\[\sum_{j_{1}, j_{2} \ldots, j_{t} \in \kappa} H_{i_{1} j_{1}} H_{i_{2} j_{2}} \ldots H_{i_{t} j_{t}} p_{j_{1} j_{2} \ldots j_{t}},\]
as presented by \citepair{Sumner2017}. Fortunately, the products of $H$ matrix entries in the above expression can be stored in a dictionary and re-used when computing subflattenings for other splits, and even for other sequence alignments. This technique, known as `memoisation', leads to a significant speed improvement, and allows us to perform the simulations we describe in \Cref{sec:analysis}.

\subsection{Number of splits}
As we have discussed, sparse matrix methods and the reduced size of subflattenings help to overcome memory constraints for computing split scores. In terms of time constraints, the sequence length is of little concern, as the time taken to construct a subflattening matrix, for example, increases only linearly as the sequence length increases. The real computational challenge which limits the usefulness of rank-based methods as a tool for inference is the number of splits, given by $2^{(n-1)}\!-\!1$, where $n$ is the number of taxa. A na\"ive way to infer a tree using the flattening or subflattening split scores is to compute a split score for every possible split, order them from smallest to largest, and select compatible splits from the top of the list until the tree is complete. There are of course other methods to obtain a tree by constructing fewer matrices and computing fewer singular values. An example of one such algorithm is presented by \citepair{Eriksson2005}. \citeauthor{Eriksson2005}'s SVD algorithm requires the computation of only $(n-1)^2 - 3$ split scores to produce an unrooted tree, and builds the tree from the bottom-up, beginning by selecting cherries. In practice, \citeauthor{Eriksson2005}'s SVD algorithm performs poorly \cite{Allman2017}. A possible explanation is that the algorithm compares split scores for splits of different sizes, which are known to be incomparable due to the bias of the flattening split score towards less balanced splits \cite{Allman2017}.  While to our knowledge, no other methods for inferring trees based on on flattening rank for species tree splits have been proposed, is it easy to see how other algorithms can be adapted to use the split scores, for example quartet methods \cite{Grunewald2006,Ranwez2001, Reaz2014, Snir2012, Willson1999}. Such methods may also avoid the split size bias. One example of this is the software is \texttt{SVDQuartets}, which was developed by \citeauthor{Chifman2014} \cite{Chifman2014, Chifman2015} and used to reconstruct trees based on the split scores for samples of possible quartets under the coalescent model.

\section{Analysis} \label{sec:analysis}
\subsection{Distribution of split scores on a 6-taxon tree}
\begin{figure}
    \centering
    \includegraphics[width=250pt]{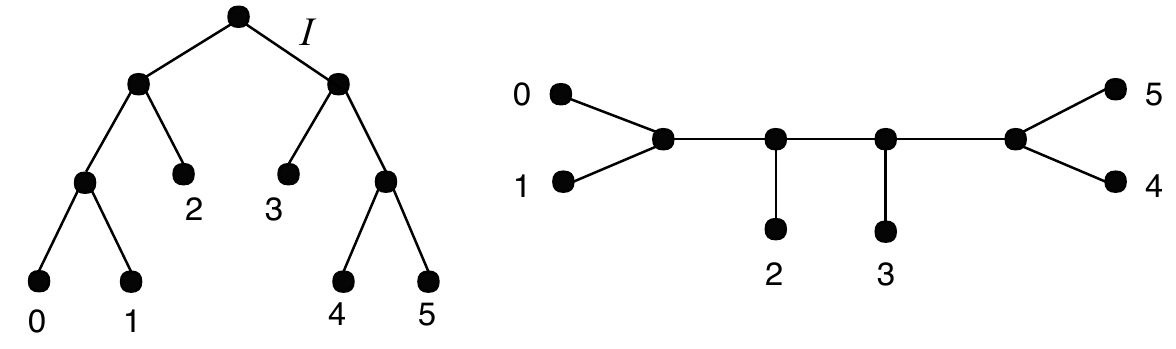}
    \caption{The 6 taxon tree used to simulate site-pattern probabilities and compute split scores for \Cref{fig:histograms}. All non-zero branch lengths are 0.1}
    \label{fig:tree6}
\end{figure}
We simulated site-pattern frequencies on the six-taxon tree shown in \Cref{fig:tree6}---the balanced, unrooted 6-taxon tree, under the Jukes Cantor model. For the first simulation, we set all branch lengths equal to 0.1. We first computed exact site-pattern probabilities and then simulated 1000 sequence alignments by drawing from the appropriate multinomial distribution. Split scores were then computed from these empirical site pattern probability distributions. Since all the branch lengths are the same, we can avoid computing split-scores unnecessarily by only computing scores for 9 inequivalent splits. That is, the chosen splits were representatives from the 9 orbits of the action of the tree automorphism group on the set of all splits. The automorphism group consists of all permutations of the leaf nodes which leave the tree unchanged. For example, on a quartet tree with cherries (1,2) and (3,4), the automorphism group is the dihedral group $\mathcal{D}_4$, consisting of permutations which swap nodes within each cherry, (12) and (34), the permutations which swaps the two cherries (13)(24) and (14)(23), as well as products of these. The interested reader can verify that the balanced 6-taxon tree shown in \Cref{fig:tree6} also has automorphism group $\mathcal{D}_4$, and the orbits of the group action are as follows:
\begin{align*}
    & [012|345] = \{012|345\},\\
	& [01|2345] = \{01|2345,\ 0123|45\},\\
	& [014|235] = \{014|235,\ 015|234,\  045|123,\ 023|145\},\\
	& [024|135] = \{024|135,\ 025|134,\  034|125,\ 035|124\},\\
	& [02|1345] = \{02|1345,\  0345|12,\ 0125|34,\ 0124|35\},\\
	& [03|1245] = \{03|1245,\  0245|13,\ 0134|25,\ 0135|24\},\\
	& [04|1235] = \{04|1235,\  05|1234,\ 0235|14,\ 0234|15\},\\
	& [0145|23] = \{0145|23\},\ \\
	& [013|245] = \{013|245\}.
\end{align*}
As an example, the two splits $014|235$ and $015|234$ are in the same class, since we can obtain one from the other by swapping the labels of the nodes $4$ and $5$, and this swap leaves the tree unchanged because these nodes form a cherry and have equal branch lengths.

Histograms were generated from split scores computed from simulated sequence lengths of 1,000bp and 10,000bp for both flattenings and subflattenings. In each case, 1,000 trials were conducted. That is, we performed 1,000 draws from the multinomial distribution in each case, and for each draw computed split scores for each of the splits shown. The entire simulation was then repeated with the same tree, but with every edge length changed from 0.1 to 0.8, emulating a tree which is much more difficult to infer. Results are displayed and discussed in \Cref{sec:histograms}.
    
\subsection{Sliding window analysis}
\label{sec:sliding_window_analysis}

The next analysis we conducted was a repeat of the sliding window analysis described in \cite{Allman2017} but for subflattenings. We also repeated the analysis using a third method for quartet reconstruction, the squangles (see work by \citeauthor{sumner2005entanglement} and \citepair{Holland2012} for background on this particular method). The sliding window analysis computes the best split (out of three possible) according to flattenings, subflattenings and squangles for each window in a four-taxon mosquito data set (\citepair{Wen2016}). Windows are 10,000bp long and at each step slide forwards 1,000bp in the alignment. The complete alignment is  over 37 million bp, giving approximately 3,700 windows in total, as described by \citeauthor{Allman2017}. This allows us to see how similar results were between the three methods across the length of the genome.
 
\subsection{Split balance bias analysis} \label{sec:split_balance}
In their \citeyear{Allman2017} paper, \citepair{Allman2017} discuss the bias of flattenings towards less `balanced' splits. Here more `balanced' splits are splits $A|B$ where $|A|$ is closer to $|B|$. For brevity, we refer to splits of size $m$ as $m$-splits. According to \citepair{Allman2017}, this bias is due to the dimension of the space of all matrices of a given shape compared to the subset of these matrices with rank $k$. Due to this bias, \citepair{Allman2017} recommends to the reader that splits of different balance not be compared by means of the split score. We hypothesised that due to their reduced dimensions, subflattenings may exhibit this bias to a lesser extent. To evaluate this, we evaluated a random selection of 190 $m$-split scores for each size $m$ on the 20 taxon tree shown in \Cref{fig:tree20}. This is the same tree used by \citepair{Allman2017} to show the existence of split-balance bias in flattening matrices. If the subflattenings are indeed less impacted by split-balance bias, we would expect the average split score to increase less dramatically for increasing split balance, when compared to split scores for flattenings.

 \begin{figure}
        \centering
        \includegraphics[width=0.85\textwidth]{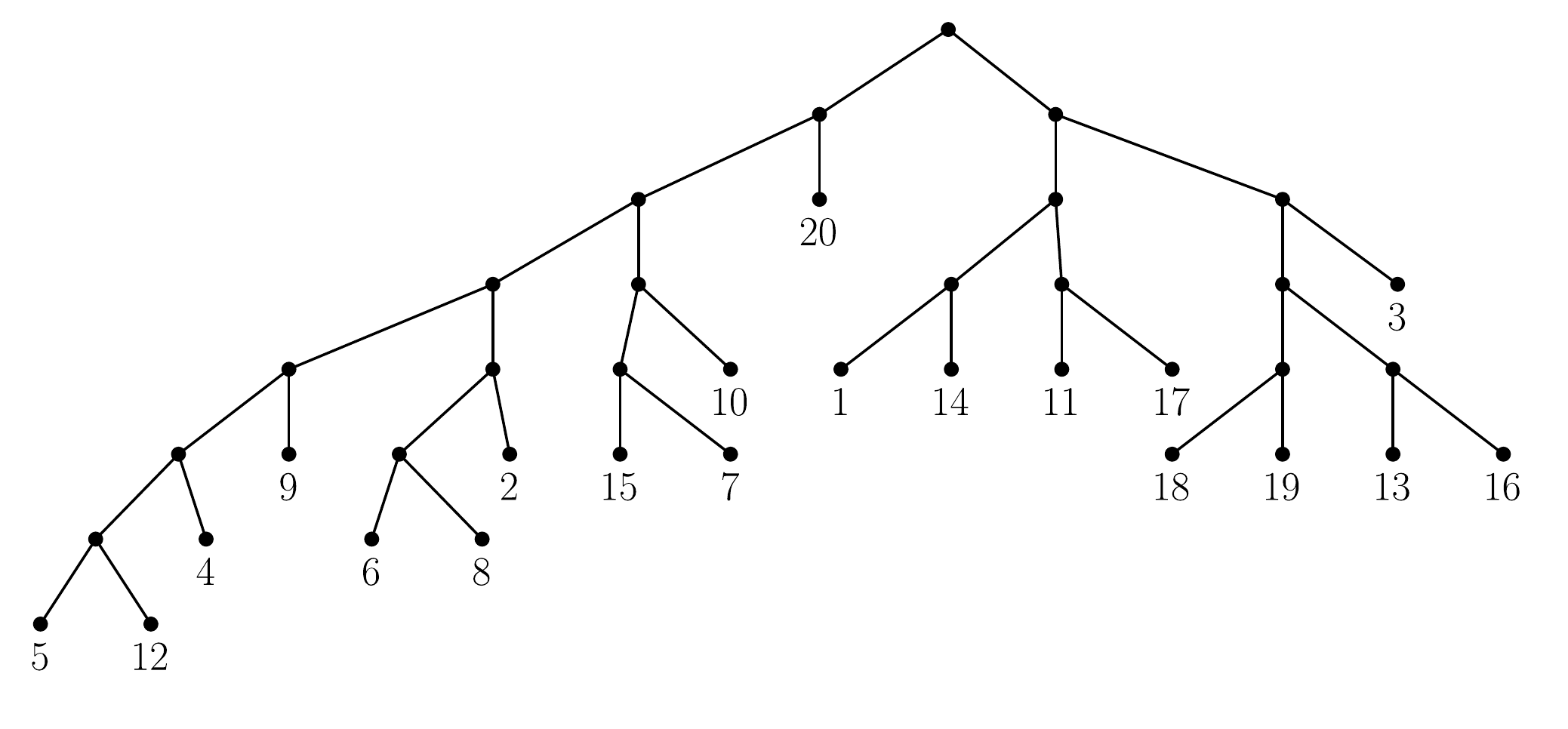}
        \caption{The 20-taxon tree used for split balance analysis described in \Cref{sec:split_balance} and by \citeauthor{Allman2017} in their \citeyear{Allman2017} paper \cite{Allman2017}.}
        \label{fig:tree20}
    \end{figure}
    
\subsection{LBA bias analysis} \label{sec:lba}

\begin{figure}
    \centering
    \includegraphics[width=170pt]{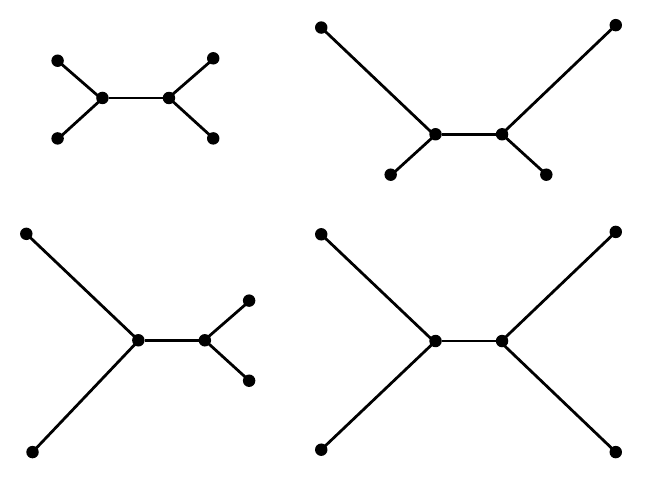}
    \caption{Quartet trees used for the simulations performed in \Cref{sec:lba}. Branch lengths are not to scale.}
    \label{fig:quartets}
\end{figure}

The next set of simulations we performed were on quartet trees with short and long branches (see \Cref{fig:quartets}). The short and long branch lengths chosen were 0.05 and 0.5 respectively. The simulation was then repeated with short and long branch lengths of 0.1 and 1.0 (maintaining the ratio of short to long branch length). In both cases, the internal branch was assigned the shorter branch length. We computed site pattern probabilities on each of these trees and drew from the corresponding multinomial distributions, simulating sequence lengths in intervals of 100bp from 100bp to 1,000bp. In each case, we perform 100 iterations and determine the number of times flattenings, subflattenings and squangles choose the true split. These simulations were designed to give an indication of how susceptible each method is to \textit{long branch attraction}: an important bias affecting some tree construction methods, first described by \citeauthor{Felsenstein1978} in \citeyear{Felsenstein1978} \cite{Felsenstein1978}. We also performed an additional simulation on a \textit{star} tree---that is, a quartet tree with an internal branch length of zero---with two long branches and two short branches. A method which is completely unbiased to long branch attraction would theoretically choose each of the three possible splits a third of the time. A biased method would choose the split which pairs the two long edges more or less often than the other two splits. Finally, we performed a similar simulation, but instead fixed the sequence length at 1000bp and allowed the long branch lengths to vary from 0.1 to 1.0 substitutions per site with the two other branch lengths fixed at 0.1. Results of these simulations are provided in \Cref{sec:lba_}.
\section{Results} \label{sec:results}
\subsection{Distribution of split scores on a 6-taxon tree } \label{sec:histograms}
    
Results of our initial simulations are shown in \Cref{fig:histograms}. We saw that split score distributions appeared to be similar between flattenings and subflattenings, and were affected similarly by increasing the sequence length. Split scores appeared to fall into four distinct groups, which we know are not a result of tree symmetries due to the chosen splits all being symmetrically distinct. Split scores appeared to be grouped by the number of \textit{shared edges} for the split on the tree. That is, the number of edges in the tree that are shared between the two subtrees induced by the two subsets of taxa. For example, the split $014|235$ on the unrooted 6-taxon caterpillar tree induces the following two subtrees (represented by the solid and dashed lines), which share two edges:
\begin{figure}[H]
    \centering
    \[\includegraphics[width=175pt]{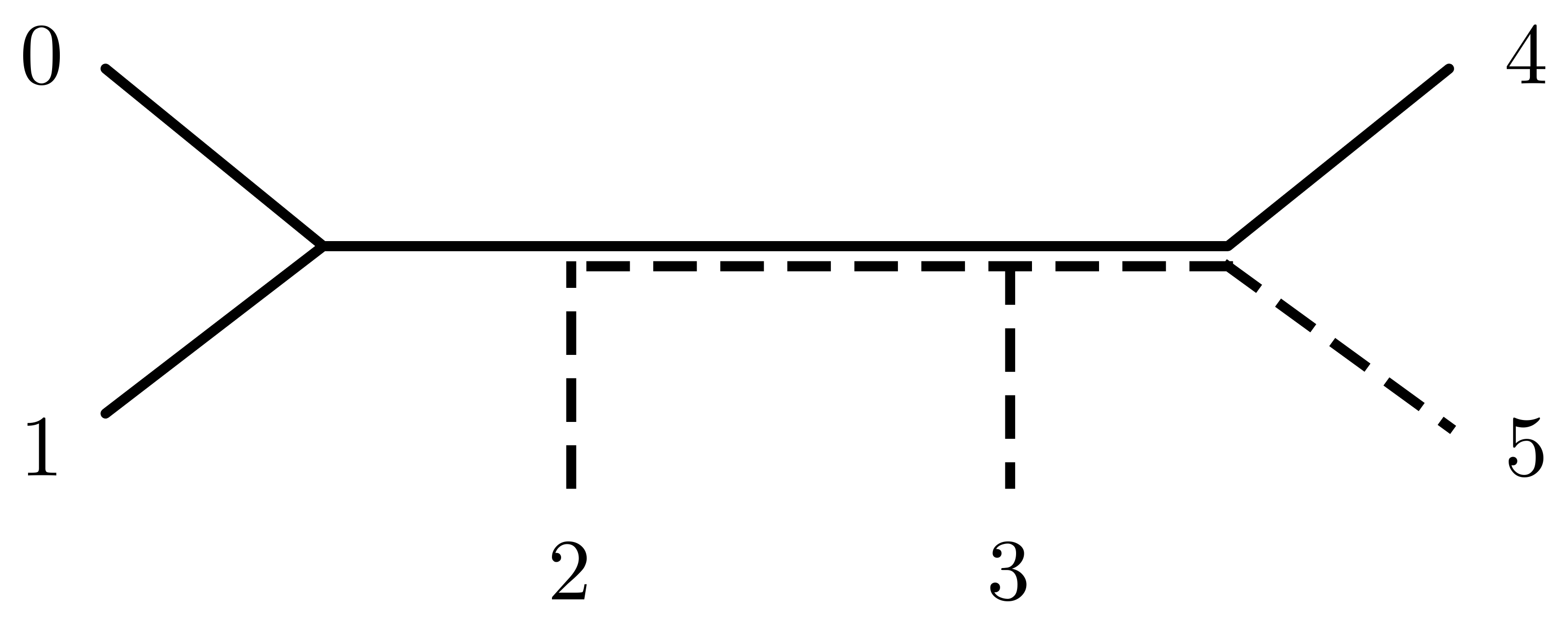}\]
    \caption{Shared edges example for the split $014|235$.}
    \label{fig:my_label}
\end{figure}

In our simulations, we commonly saw that the number of shared edges between subtrees induced by a split appeared to explain the ordering and grouping of split scores. In fact, despite the connection between the flattening/subflattening rank and the parsimony score for a split, the number of shared edges is a much better predictor of the split score than the parsimony score. As a simple test, we took the tree in \Cref{fig:tree6} and drew each edge length uniformly between 0.1 and 1 (Jukes-Cantor), and computed exact site-pattern probabilities and split scores for subflattenings corresponding to each split. Repeating this 20 times and compiling the scores, we fitted a linear model to predict split score based on shared edges, parsimony score and split balance. The best model (according to adjusted $R^2$) used only the number of shared edges, and explained approximately $61\%$ of the variability in the split score. This outcome was similar for flattenings, and increasing the lower bound on the edge lengths to 0.5 increased the $R^2$ value to approximately $89\%$.

\begin{figure}[!htb]
    \centering
    \includegraphics[width=\textwidth]{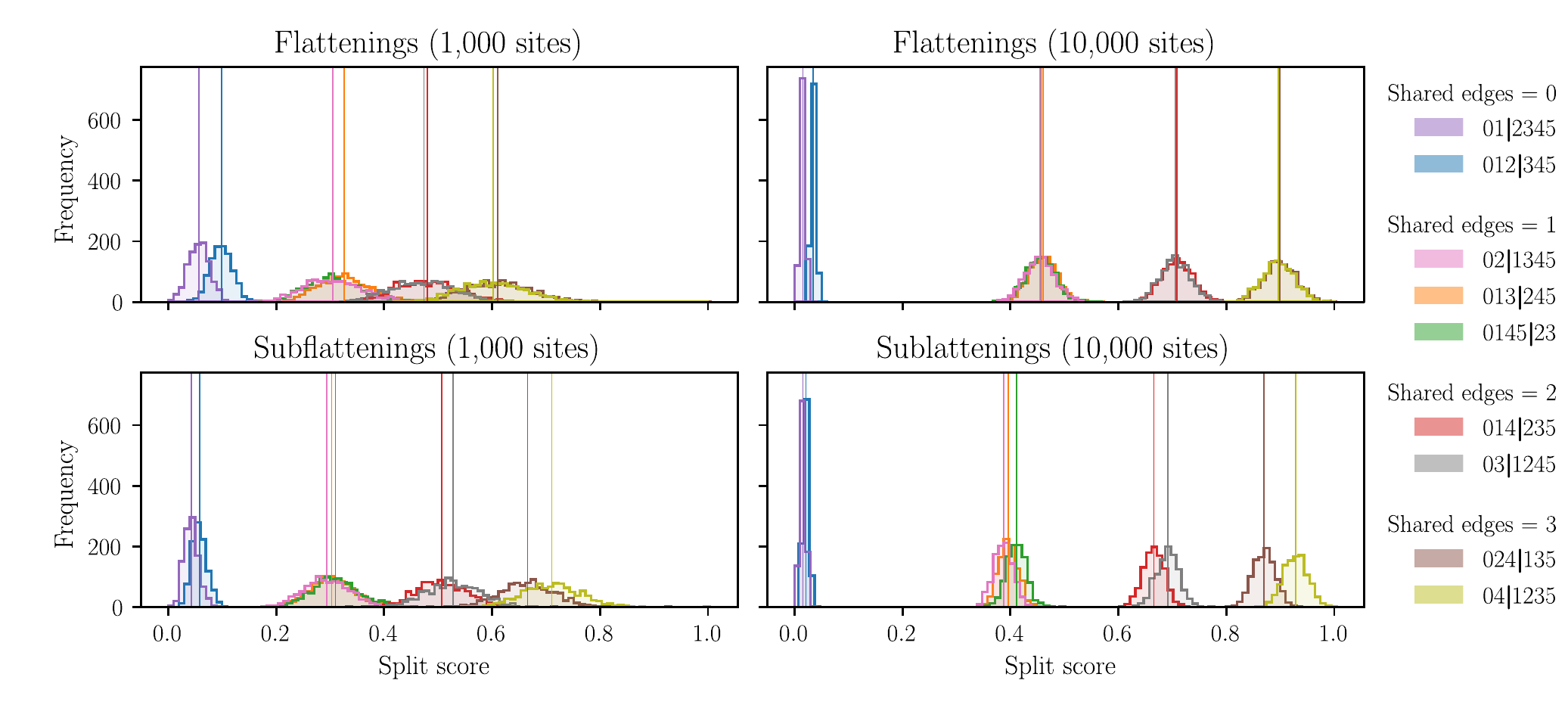}
    \caption{Distribution of split scores for flattenings and subflattenings on simulated pattern frequencies on the 6-taxon tree shown in \Cref{fig:tree6}. Branch lengths are all set to 0.1. The distributions are clearly separated by the number of shared edges.}
    \label{fig:histograms}
\end{figure}

While, like the split parsimony score, the number of shared edges is unhelpful in discovering true splits (since the true tree must be known in order to compute it), the measure may be helpful in determining exactly which properties of the tree most effect the split score, and should be investigated further.

\begin{figure}[!htb]
    \centering
    \includegraphics[width=\textwidth]{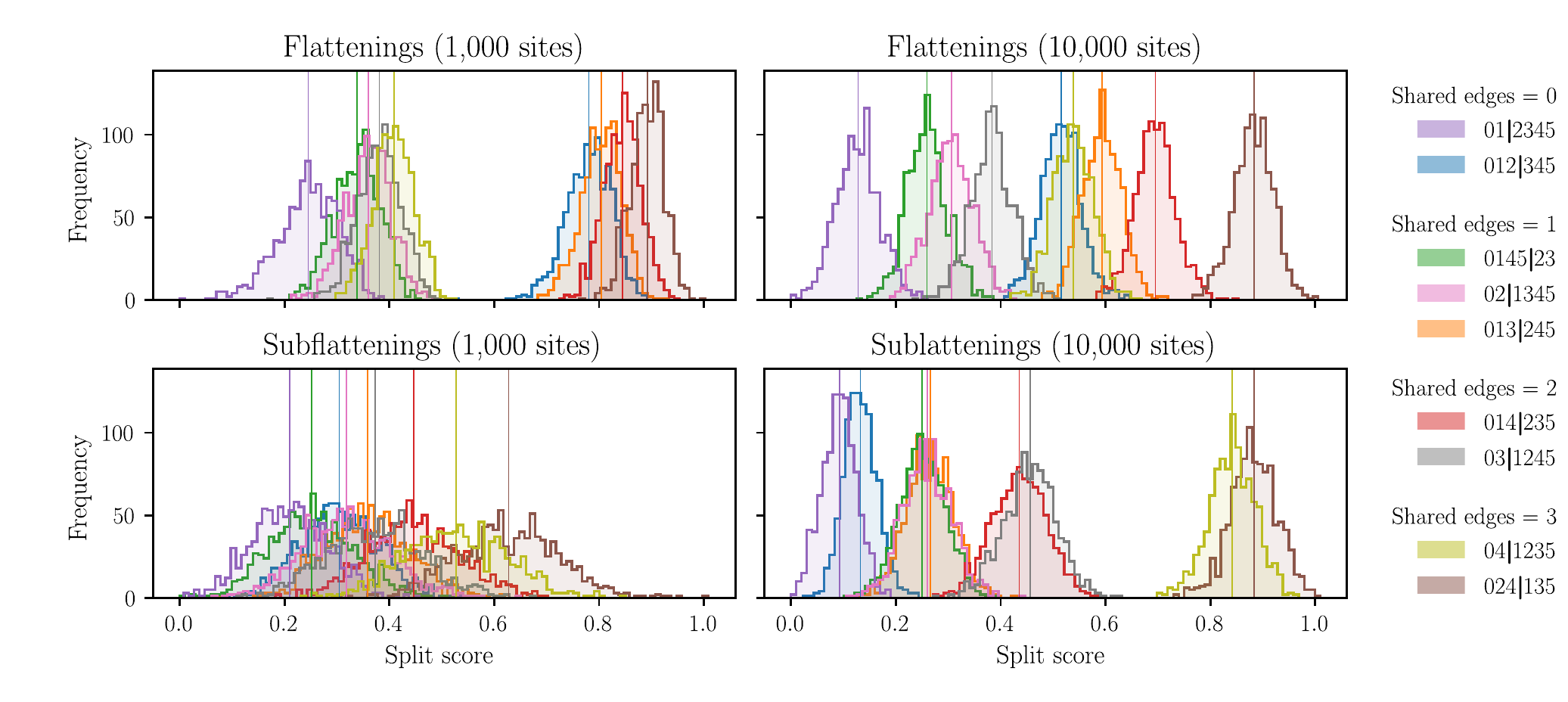}
    \caption{Distribution of split scores for flattenings and subflattenings on simulated pattern frequencies on the 6-taxon tree shown in \Cref{fig:tree6}. Branch lengths are all set to 0.8.}
    \label{fig:histograms_2}
\end{figure}

For the second simulation, the Jukes Cantor branch length for every branch in the tree were changed from 0.1 to 0.8. This change should make the tree more difficult to infer, since nucleotide substitutions are much more likely to occur along each edge. The results are shown in \Cref{fig:histograms_2}. We will discuss the flattening scores first. We see that for the more difficult tree and a sequence length of 1,000, the flattening split scores are grouped by the split size, and then ordered by the number of shared edges. This seems to suggest that split score bias worsens with a weaker phylogenetic signal. Increasing the sequence length to 10,000, we see that the effect of split size has a lesser impact on the distribution of scores, however the true split $012|345$ of size 3 still produced scores which were higher than three false splits of size 2. The subflattenings with a sequence length of 1,000bp seemed to perform similarly to the flattenings with sequence length 10,000bp, but with scores packed more closely together. Increasing the sequence length to 10,000, the subflattenings produced scores which looked much more similar to the results from \Cref{fig:histograms}, with splits grouped by the number of shared edges, and minimal impact of split size bias.

\subsection{Sliding window analysis }

    As expected, subflattenings and squangles were able to pick up the chromosomal inversion in the genome, just as flattenings did in the sliding window analysis completed by \citepair{Allman2017}. The three methods agreed upon the most likely phylogeny for the majority of windows. Interestingly, the most common case involving some disagreement between the three methods was that in which squangles and flattenings chose a different split to the subflattenings. Further, the least common case was when squangles and subflattenings were in agreement with each other, but not with the flattenings. This is visible in \Cref{fig:sliding_window_plot} and more clearly in \Cref{fig:sliding_window_venn}. This was not expected because of the sense that squangles and subflattenings have more in common---they both involve the same transformation of site-pattern probabilities detailed in \Cref{sec:background}.
    
    \begin{figure}[!htb]
        \centering
        \includegraphics[width=\textwidth]{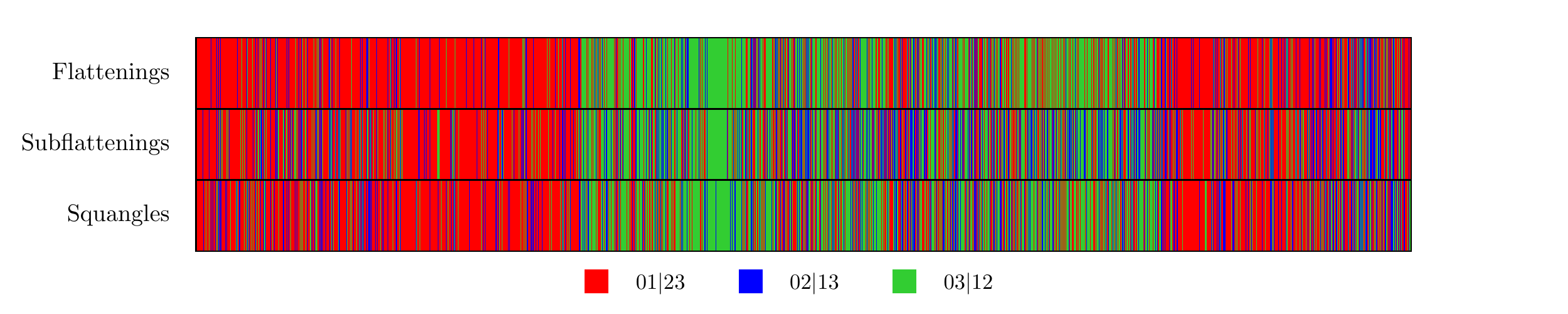}
        \caption{Results of the sliding window analysis detailed in \Cref{sec:sliding_window_analysis}. Each vertical line represents a window, and the colour represents the split which was chosen by each method for that particular window.}
        \label{fig:sliding_window_plot}
    \end{figure}
    
    \begin{figure}[!htb]
        \centering
        \includegraphics[width=300pt]{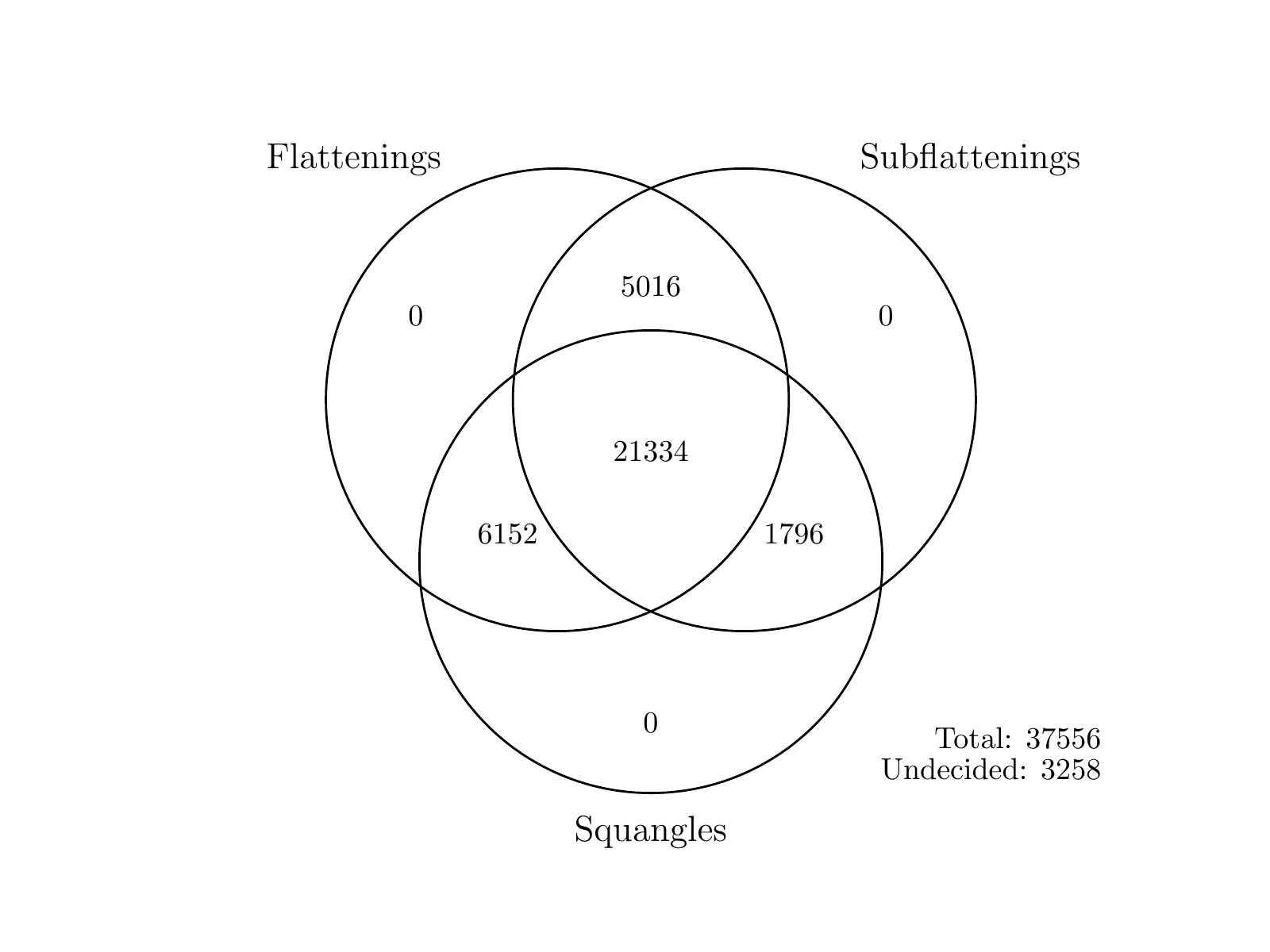}
        \caption{A summary of the results shown in \Cref{fig:sliding_window_plot} in the form of a Venn diagram. We can see that the three methods were in agreement for the majority of windows, and of the cases where there was a disagreement between two methods and a third, flattenings and squangles agreed most often.}
        \label{fig:sliding_window_venn}
    \end{figure}

\subsection{Split balance bias analysis }

From our simulations on the 20-taxon tree, both flattenings and subflattenings appeared to be biased towards less balanced splits, but as predicted, subflattenings appeared to be slightly less biased. \Cref{fig:balance_bias} shows that, compared to split scores for flattenings, split scores for subflattenings rose more slowly as split size increased. While \citeauthor{Allman2017} suggest that correction to this bias exists theoretically, they deemed it to be of little use in practice \cite{Allman2017}. In our view, the link between the rank properties of flattening/subflattening matrices and the split parsimony score suggests that a practical correction allowing comparison between splits of different sizes does not exist.

\begin{figure}[!htb]
    \centering
    \includegraphics[width=\textwidth/2]{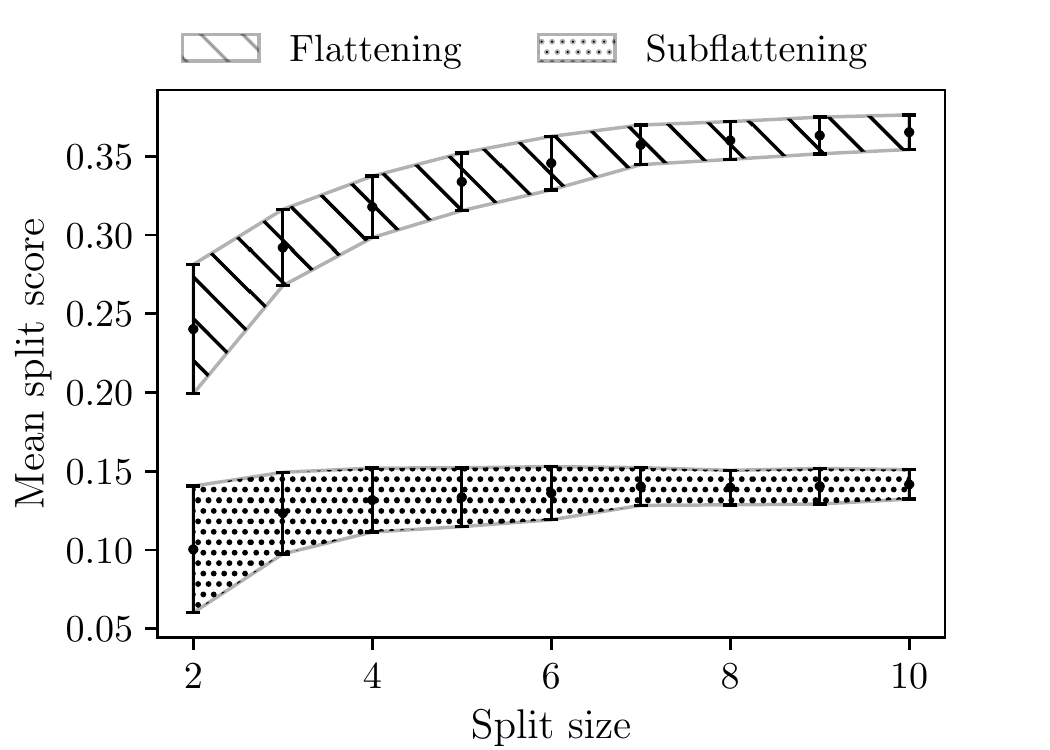}
    \hspace{-5mm}\includegraphics[width=\textwidth/2]{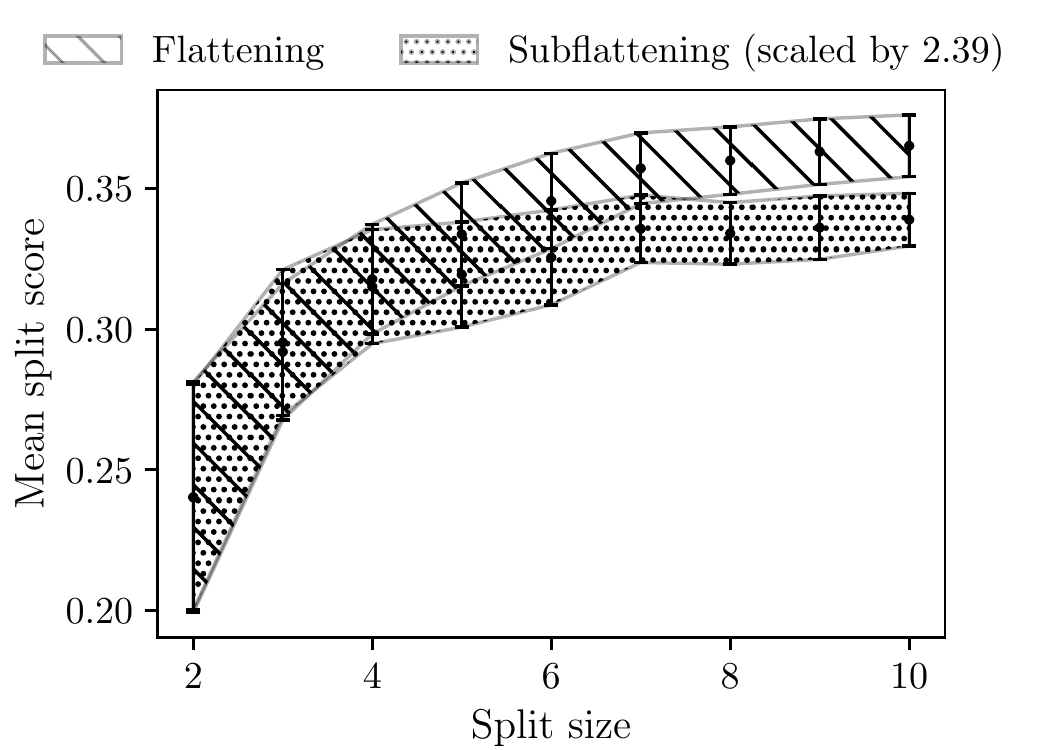}
    \caption{Split scores computed for a sample of 190 splits for each split size on the 20-taxon tree shown in \Cref{fig:tree20}, with branch lengths all set to 0.05, and sequence length set to 500bp. The plot on the right is the same as the left, but with all subflattening scores scaled so that the first points overlap. Error bars indicate sample standard deviation.}
    \label{fig:balance_bias}
\end{figure}

While we only simulated a single sequence alignment, repeated simulations yielded results which seemed to be stable. Future investigation into the effect of split balance could include a similar simulation study, but with a systematic variation of sequence length, branch lengths and/or tree topology.

\subsection{LBA bias analysis } \label{sec:lba_}
 
\begin{figure}[!htb]
    \centering
    \includegraphics[width=0.95\textwidth]{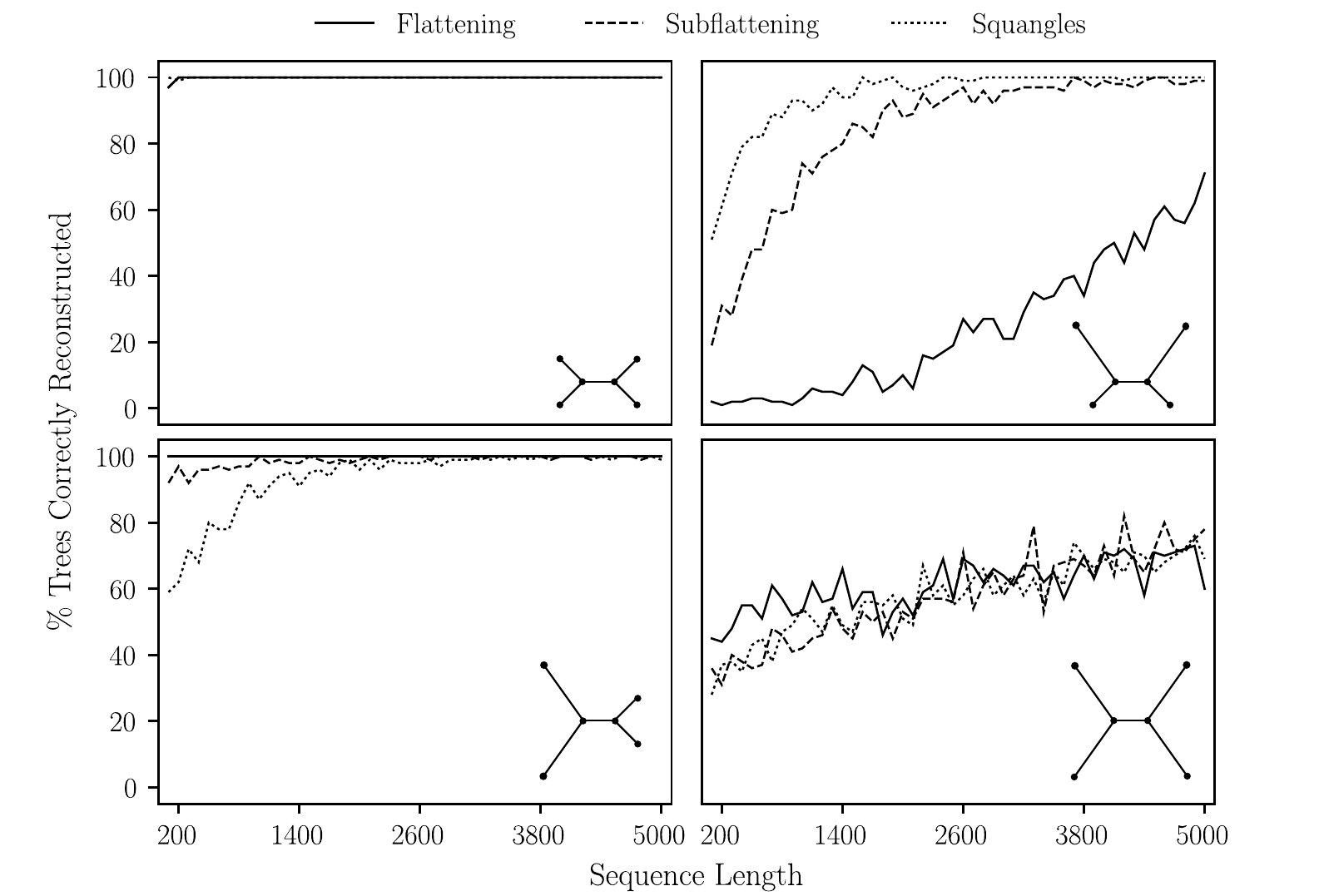}
    \caption{The percentages of times each method was able to choose the correct split from the simulated pattern counts. For each sequence length, we generated 100 different sequence alignments. For each of the four quartet trees, short edges have Jukes-Cantor length 0.05 and long edges have length 0.5.}
    \label{fig:lba_bias}
\end{figure}
    
The results of our simulations on quartet trees are shown in \Cref{fig:lba_bias}. As expected, we observe flattenings, subflattenings and squangles performing similarly on the trees in which all branches were either short or long, with worse over-all performance in latter case. Flattenings showed significantly worse performance on the top-right quartet---an indication that flattenings were incorrectly pairing the long edges together far more often than subflattenings were. Further, flattenings performed slightly better for the bottom left quartet, indicating that the flattenings are more confident in correctly joining the two long branches. These results might indicate that subflattenings and squangles are less impacted by long branch attraction bias when compared to flattenings. Comparing subflattenings to squangles, we saw the squangles correctly reconstruct more trees for the top-right quartet, and fewer trees for the bottom-left quartet.

\begin{figure}[!htb]
    \centering
    \includegraphics[width=0.95\textwidth]{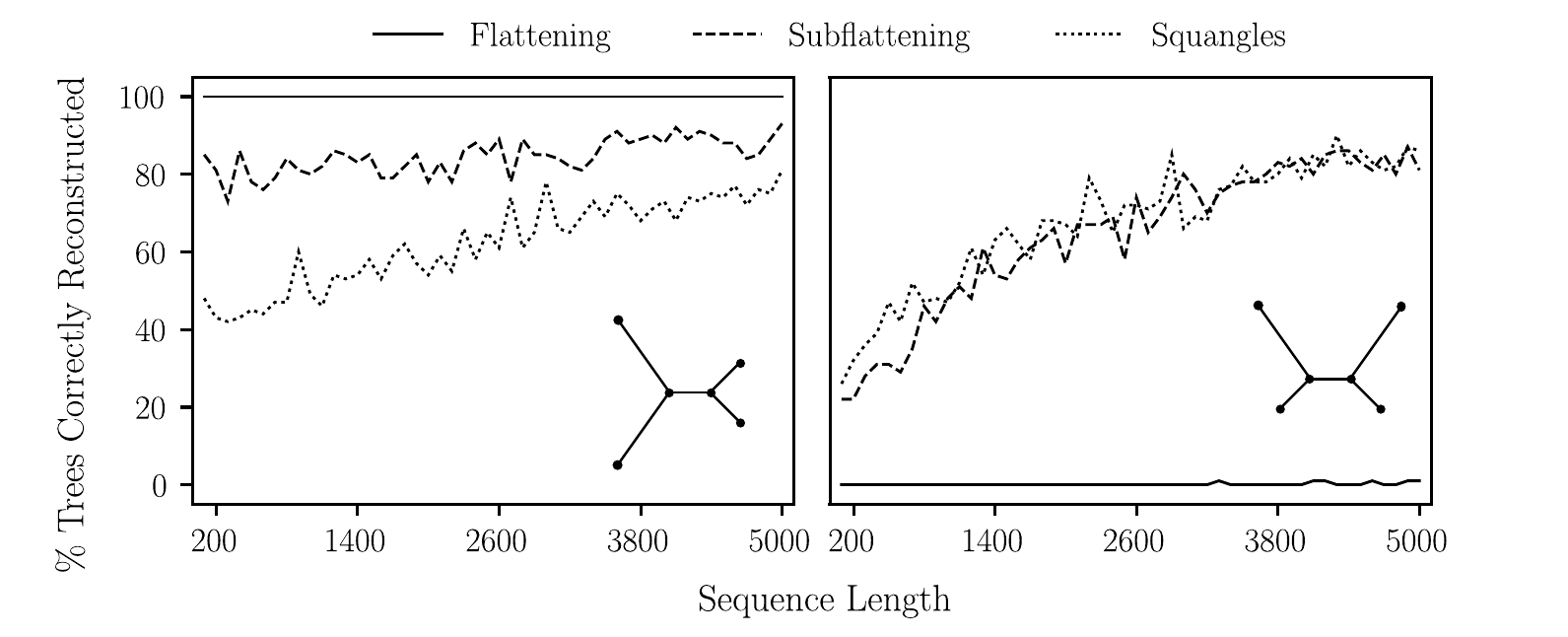}
    \caption{The percentages of times each method was able to choose the correct split from the simulated pattern counts. For each sequence length, we generated 100 different sequence alignments. For each of the four quartet trees, short edges have Jukes-Cantor length 0.05 and long edges have length 0.5.}
    \label{fig:lba_bias_longer}
\end{figure}

We repeated the simulations on the top-right and bottom-left quartets but with longer branches (short branch length set to 0.1 and long branch length set to 1.0), and the results are shown in \Cref{fig:lba_bias_longer}. We see that the difference between the flattening and subflattening performance when presented with two long branches was even more pronounced. We note that the long branch length of 1.0 means that substitutions along the long branches are more likely to occur than not, making these trees much more difficult to infer.

\begin{figure}[!htb]
    \centering
    \includegraphics[width=0.9\textwidth]{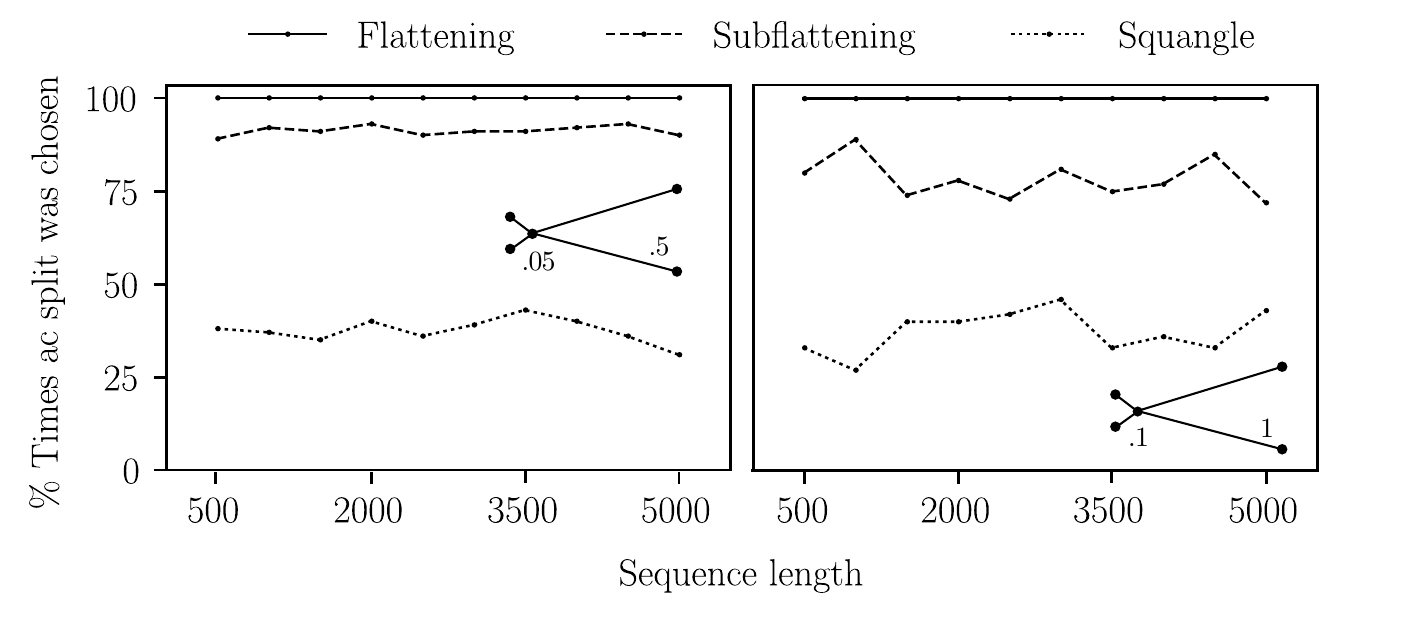}
    \caption{The percentages of times each method chose the split $ac|bd$ on the star tree on leaves $\{a,b,c,d\}$ with an internal edge length of zero, a branch length of 0.1 for leaves $b$ and $d$ and longer edges for leaves $a$ and $c$. For each sequence length, 100 sequence alignments were generated. Branch lengths are as shown.}
    \label{fig:lba_bias_star}
\end{figure}

\Cref{fig:lba_bias_star} shows the percentage of times flattenings and subflattenings correctly reconstructed the star tree with two long branches and an internal branch length of zero. While flattenings paired the two long edges almost every time, subflattenings did so less often. This may further indicate that subflattenings are less biased to long branch attraction, though not completely unbiased, since a totally unbiased method would pair the long edges approximately a third of the time. Increasing all four branch lengths (see the right plot in \Cref{fig:lba_bias_star}) seemed to increase the performance of the subflattenings, and \Cref{fig:lba_bias_star_increasing} seems to suggest that this effect is due to the increased length of the short branches.

\begin{figure}[!htb]
    \centering
    \includegraphics[width=0.95\textwidth]{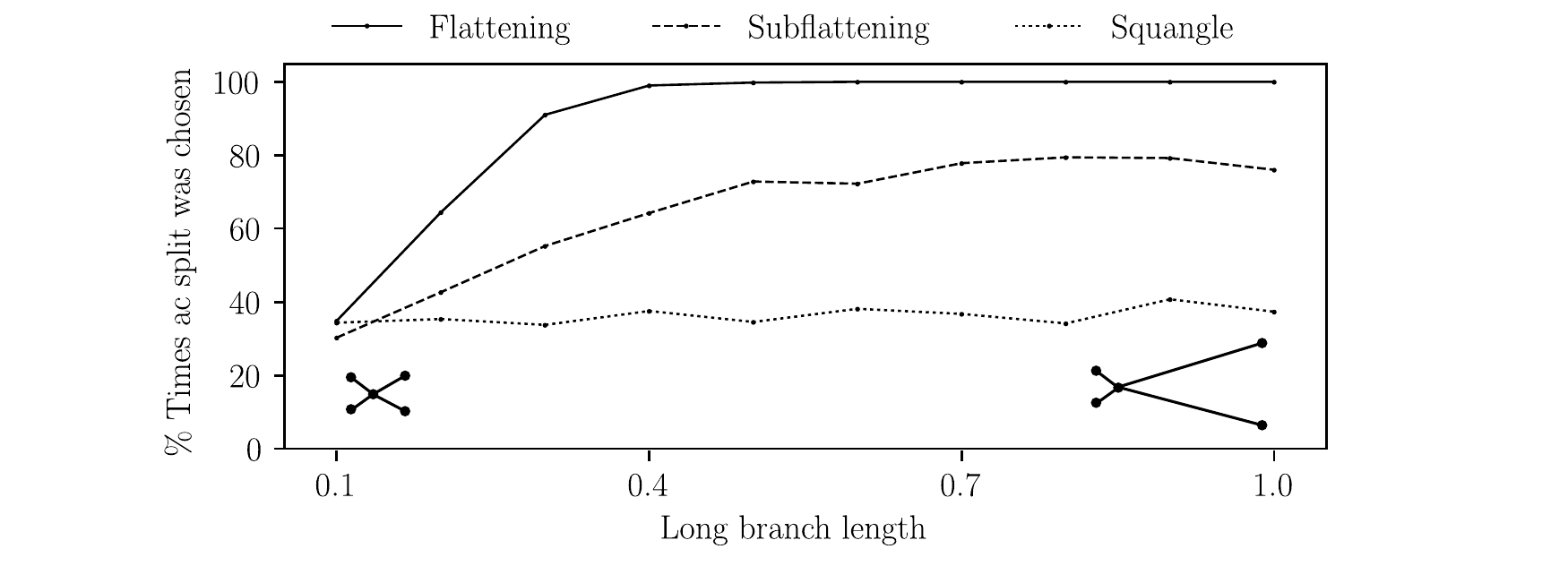}
    \caption{The percentages of times each method chose the split $ac|bd$ on the star tree on leaves $\{a,b,c,d\}$ with an internal edge length of zero, two short branches of length 0.1 for leaves $b$ and $d$, and varying branch lengths for leaves $a$ and $c$. For each increase of the long branch length, 100 sequence alignments of length 1000 were generated.}
    \label{fig:lba_bias_star_increasing}
\end{figure}

Interestingly, the squangles method appeared to be the least affected by long branch attraction bias in these simulations. A theoretical investigation is recommended in order to confirm and make sense of these observations.

\section{Discussion} \label{sec:discussion}
Flattenings and subflattenings, and more generally split and rank based tools, encompass some interesting algebraic and statistical ideas, and motivate methods for phylogenetic inference. We have presented results which show that split scores computed from subflattenings are comparable to those computed using flattenings, supporting the use of subflattenings as an alternative tool to flattenings for phylogenetic inference. We also give some evidence that subflattenings may be less impacted by certain biases, namely those arising from split size/balance and long branch attraction. We think this is worth exploring further, but note that some split size bias is likely unavoidable with methods relating to flattenings, due to the relationship between the flattening matrix rank and the parsimony score. It is our hope that this work will contribute to an improved understanding of rank-based methods for constructing phylogenetic trees. There are a number of other avenues available for further research in this area. While we have looked at reconstructing quartet trees in this paper, algorithms that utilise flattenings to reconstruct larger trees have been developed \cite{Eriksson2005} and can instead utilise subflattenings without needing to be adapted in any way. It would be beneficial to systematically compare tree reconstruction methods involving flattenings and subflattenings to various other approaches.

\addcontentsline{toc}{section}{References}
\small
\bibliographystyle{abbrvnat}
\bibliography{subflattenings}

\end{document}